\documentclass[aps,pra,twocolumn,superscriptaddress]{revtex4}

\usepackage{graphicx,enumerate,verbatim,bbold}
\usepackage{amsmath,amssymb,amsthm,float,mathrsfs}
\usepackage[colorlinks]{hyperref}
\usepackage{cstyle}[2013/11/07]

\newcommand{\dani}[1]{{\color{black} #1}}
\newcommand{\colin}[1]{#1}

\makeatletter
\def\namedlabel#1#2{\begingroup
   \def\@currentlabel{#2}%
   \phantomsection\label{#1}\endgroup
}
\makeatother
\newcommand{\reals}{\mathbb{R}}

\newcommand{\kb}[2]{\mbox{$|#1 \rangle \langle #2 |$}}

\newcommand{\comm}[1]{}

\newcommand{\subfigref}[2][{}]{\hyperref[#2]{\ref{#2}#1}}

\begin{document}
\title{Analysis of a proposal for a realistic loophole-free Bell test with atom-light entanglement}
\author{C. Teo}
\affiliation{Centre for Quantum Technologies, National University of Singapore, Singapore}
\author{J. Min\'{a}\v{r}}
\affiliation{Centre for Quantum Technologies, National University of Singapore, Singapore}
\author{D. Cavalcanti}
\affiliation{Centre for Quantum Technologies, National University of Singapore, Singapore}
\affiliation{ICFO-Institut de Ciencies Fotoniques, 08860 Castelldefels, Barcelona, Spain}
\author{V. Scarani}
\affiliation{Centre for Quantum Technologies, National University of Singapore, Singapore}
\affiliation{Department of Physics, National University of Singapore, Singapore}

\begin{abstract}
\dani{The violation of Bell inequalities where both detection and locality loopholes are closed is crucial for device independent assessments of quantum information. While of technological nature, the simultaneous closing of both loopholes still remains a challenge. In Ref.~\cite{TeoAra13}, a realistic setup to produce an atom-photon entangled state that could reach a loophole free Bell inequality violation within current experimental technology \colin{was proposed}. Here we improve the analysis of this proposal by giving an analytical treatment that shows that}
\colin{the state proposed in Ref.~\cite{TeoAra13} can violate a Bell inequality for arbitrarily low photodectection efficiency, when all other losses are ignored. Moreover, it is also able to violate a Bell inequality considering only atomic and homodyne measurements eliminating the need to consider inefficient photocounting measurements. In this case, the maximum Clauser-Horne-Shimony-Holt (CHSH) inequality violation achievable is 2.29, and the minimum transmission required for violation is about 68\%. Finally, we show that by postselecting on an atomic measurement, one can engineer superpositions of coherent states for various coherent state amplitudes.}
\end{abstract}
\maketitle
\section{Introduction}

In recent years, it has been realized that nonlocal correlations can be certified without any additional assumptions on the internal mechanism of the devices used in the experiment. These nonlocal correlations can be obtained from entangled quantum systems by measuring them in appropriately chosen local observables. This is called a Bell test since the nonlocal nature of the measurement outcomes can be certified by the violation of certain constraints known as Bell inequalities \cite{Bell,Review}.

Once the existence of such correlations is established, they can be used in what is now referred to as \textit{device-independent} protocols \cite{diqkd,dirandom, Masanes11,distateest,rabelo2011device}.

Many Bell tests have been performed (see Sec. VII in \cite{Review}), most of them involving atomic and photonic systems (Bell tests in the realm of high energy physics have also been proposed and carried out, see for example ~\cite{KaonRef, KaonLPRef} and references therein). But, all experiments performed so far have suffered either from the detection loophole or the locality loophole \cite{aspect_experimental_1982,weihs_violation_1998,scheidl_violation_2010,matsukevich_bell_2008, rowe_experimental_2001, ansmann_violation_2009}. Needless to say, both the ultimate confirmation of the nonlocal character of nature and the implementations of device-independent quantum information protocols rely on loophole-free Bell tests.

The locality loophole has been closed in experiments with entangled propagating photons \cite{aspect_experimental_1982,weihs_violation_1998,scheidl_violation_2010}, while the detection loophole has been closed with stationary systems like atoms, ions, circuits \cite{matsukevich_bell_2008, rowe_experimental_2001, ansmann_violation_2009} and only very recently with photons, though not yet in separate locations \cite{GiuMec12,ChrMcC13}. It seems natural to seek for the best of both worlds by using a hybrid scenario. The main technological challenge is to have both fast and efficient detection at the same time. Typically, the detection of stationary information carriers (atoms, ions, \ldots) is very efficient, but relatively slow. This implies that the other information carrier (photon) must propagate a very large distance, and the losses add to the overall inefficiency of detection. On the other hand, it is still very challenging to have very efficient and yet still fast photon detection at the single photon level.

In Ref.~\cite{TeoAra13}, we \dani{and collaborators} proposed a Bell test which requires very low photon detection efficiency and a possible experimental implementation, feasible with available technology. It uses a single atom coupled to the field of a cavity in order to produce a specific entangled state between the atom and the light emitted by the cavity. The idea is to combine efficient detection schemes on the atomic side and the coherent nature of a light field to perform the detection on the photonic side \cite{cavalcanti10,quintino12}. This technique significantly reduces the required photon detection efficiency and is a promising candidate for a loophole-free experiment.

Here, we present a more detailed analysis of the measurement scenario of that proposal and show that 
the state proposed allows a violation of the CHSH inequality with only atomic and homodyne measurements, eliminating the need to perform inefficient photoncounting. The paper is structured as follows. In section \ref{sec:measurement_scenario}, we introduce the target state used in the Bell test. We show that for the measurement scenario of Ref.~\cite{TeoAra13}, we can achieve violation for arbitrarily low detection efficiency. Bell tests on the state using only atomic and homodyne measurements is also considered, and the possibility of inefficient atomic detection is discussed. We then show a possible route to producing entangled coherent states, and conclude in section \ref{sec:conclude}.

\section{Bell test scenario} \label{sec:measurement_scenario}

\dani{We will consider a typical Bell test involving two parts, A and B, each measuring two possible observables ($A_0, A_1$ and $B_0, B_1$ respectively) with dichotomic outcomes $\pm1$. In this scenario the only relevant Bell inequality is the 
\colin{CHSH} inequality given by
\begin{equation}
  \chsh = \mean{A_0 B_1} + \mean{A_1 B_1} + \mean{A_0 B_0} - \mean{A_1 B_0} \leq 2. \label{eqn:CHSH}
\end{equation}
The violation of this inequality implies that the measured statistics \colin{are} nonlocal. In the quantum case, these average values are formally given by the standard rule
\begin{equation}
 \mean{A_i B_j}=\Tr(A_i\otimes B_j \rho),
\end{equation}
where $\rho$ is a quantum state and $A_i$ and $B_j$ Hermitian operators with $\pm 1$ eigenvalues.

In the present paper we will deal with a scenario in which part A holds an atomic system and part B a photonic system. In what follows we will specify the state and measurements that will lead to the desired Bell violation of the CHSH inequality.}

\subsection{Quantum state}

In this article, we are interested in the \dani{following} state
\begin{equation}
  \ket{\psi} = \cos\nu \ket{s,0} + \sin\nu \ket{g,\alpha}, \label{eqn:target_state}
\end{equation}
where $\ket{s}$ and $\ket{g}$ are two energy levels of an atom, and $\ket{0}$ and $\ket{\alpha}$ are the vacuum and the coherent state of the electromagnetic field, with amplitude $\alpha$, respectively. This state is a hybrid entangled state of atom-field.

\subsection{Local measurements}\label{loc meas}

Consider a Bell test on the entangled system with the measurement operators
\begin{align}
  A_0 &= \cos \gamma \sig_z + \sin\gamma \sig_x  \label{eqn:A0},\\
  A_1 &= \cos \gamma \sig_z - \sin \gamma \sig_x \label{eqn:A1},
\end{align}
on the atomic system, where $\sig_x$ and $\sigma_z$ are the usual Pauli operators. \dani{For the photonic measurements,} as we would like to use the CHSH inequality, we will need a suitable dichotomization of the infinite photon Hilbert space. We choose the two operators,
\begin{align}
  B_0 &= 2\int^b_{-b} dx\, \proj{x} - \one, \quad {\rm and}\\
  B_1 &= 2\proj{0} -\one \label{eqn:B1_n}.
\end{align}
$B_0$ represents homodyne measurement of the $X$ quadrature, followed by a binning that outputs ``$+1$'' when the measurement result $\in [-b,b]$ and ``$-1$'' otherwise. $B_1$ represents photocounting, with the outcome ``$+1$'', when no photon is detected and ``$-1$'' otherwise.

\subsection{Bell violation: perfect case} \label{sec:opt_methods}
With the above constraints, we now show a semi-analytical optimization procedure which simplifies numerical computations of the maximum Bell violation.
The structure of the $A$ measurements allows one to optimize directly over the measurement angle $\gamma$. To see this, substitute Eqs. \eqref{eqn:A0} and \eqref{eqn:A1}, and rewrite $\chsh$ in \eqnref{eqn:CHSH} as
\begin{align}
  \chsh &= \mean{(A_0+A_1)B_1 + (A_0-A_1)B_0} \\
  &= 2\cos\gamma\mean{\sig_z B_1} + 2\sin \gamma \mean{\sig_x B_0}.
\end{align}
The maximum of $\chsh$ over $\gamma$ is thus
\begin{equation}
  \chsh_\gamma \dani{:=\max_\gamma \chsh}= 2 \sqrt{\big(\Tr(\rho \sig_x B_0)\big)^2 + \big(\Tr(\rho \sig_z B_1) \big)^2}, \label{eqn:chsh_firstopt}
\end{equation}
where the subscript denotes the variable which has been optimized over, and $\rho = \proj{\psi}$. Inserting \eqref{eqn:target_state} for $\ket{\psi}$ and defining
\begin{align}
  c_1 &= \half 1 \Tr\big((\kb{0}{\alpha}+\kb{\alpha}{0})B_0\big), \label{eqn:c1_perf}\\
  c_2 &= \Tr(\proj{0}B_1), \\
  c_3 &= \Tr(\proj{\alpha}B_1), \label{eqn:c3_perf}
\end{align}
yields
\begin{equation}
  \chsh_\gamma = 2\sqrt{(2 c_1\cos\nu\sin\nu)^2 + (c_2 \cos^2\nu-c_3\sin^2\nu)^2}. \label{eqn:chsh_nextopt}
\end{equation}
We now optimize \eqref{eqn:chsh_nextopt} over the state parameter $\nu$. \colin{We first rewrite the above equation as
\begin{align}
  \chsh_\gamma &= 2\sqrt{A\sin^4\nu +B \sin^2\nu + c_2^2}, \label{eqn:opt_r1} \\
  &= 2\sqrt{f}
\end{align}
where $A = (c_2+c_3)^2-4 c_1^2$ and $B= 4 c_1^2-2c_2(c_2+c_3)$. We plot the behaviour of the function $f$ for the two cases: $|c_2|>|c_3|$ and $|c_2|<|c_3|$ in \figref{fig:pic_proof}. This function has 3 extrema, is symmetric about $\sin\nu = 0$, and has $|c_2| \leq 1$ and $|c_3| \leq 1$ at $\sin\nu = 0$ and $\sin\nu=\pm 1$ respectively. From \figref{fig:pic_proof}, it becomes intuitively obvious that to have $\chsh > 2$, we need $f$ to have two maxima between $\sin\nu \in (-1,1)$. }
\begin{figure}[ht]
\centering
\includegraphics[width = \columnwidth]{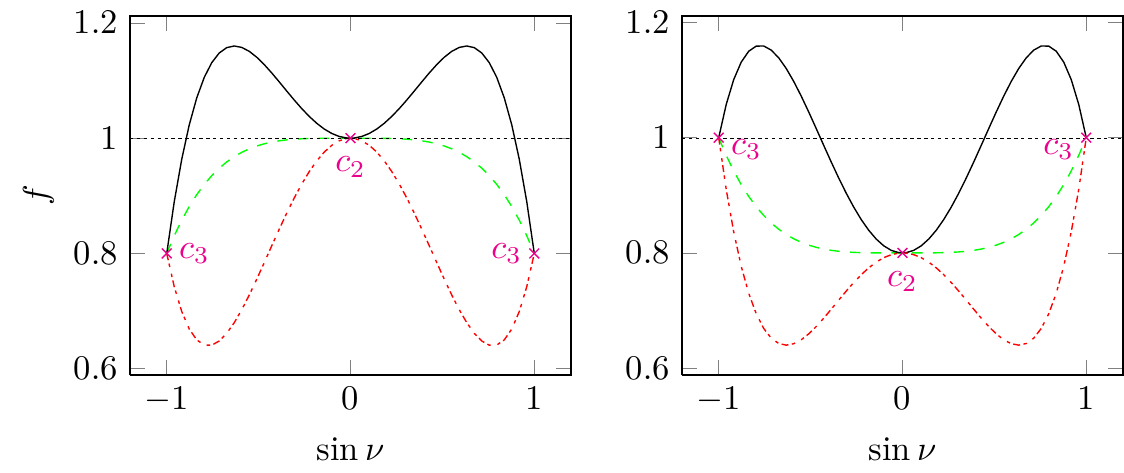}
\caption{(Color online) The above plots show the possible plots of the function $f$ given $f(\sin\nu = 0) = |c_2|$, and $f(\sin\nu = \pm 1) = |c_3| $. The left graph shows the case when $|c_2|>|c_3|$ and the right graph is when $|c_2|<|c_3|$. As is evident from these graphs, the only possible case for $f\geq1$ is when $f$ and consequently $\chsh_\gamma$ has 2 maxima in $\sin\nu\in[-1,1]$ (solid black curve). Note that in both graphs we have let the larger of $|c_2|$ and $|c_3|$ be 1. This is not necessarily the case, and depends on the specific measurements used. Thus, the condition of 2 maxima is a necessary but insufficient condition for $\chsh_\gamma >2$.}\label{fig:pic_proof}
\end{figure}
Simple differentiation shows that the extrema satisfy the condition
\begin{equation}
  (2A\sin^2\nu+B) \sin\nu = 0.
\end{equation}
Thus, the conditions for \eqref{eqn:opt_r1} to have two maxima are
\begin{align}
  0< &\frac{B}{-2A} < 1 \quad \textrm{and,}\\
  A &< 0.
\end{align}
Rewriting the above conditions in terms of $c_1,c_2$ and $c_3$ gives,
\begin{enumerate}
\centering
  \item[C1: \namedlabel{eqn:nec_con1}{C1}]   \qquad $c_3(c_2+c_3) < 2c_1^2$
  \item[C2: \namedlabel{eqn:nec_con2}{C2}] \qquad $c_2(c_2+c_3) < 2c_1^2$
\end{enumerate}
as the necessary conditions for $\chsh>2$. If either $c_2=1$ or $c_3=1$, and both conditions \eqref{eqn:nec_con1} and \eqref{eqn:nec_con2} are satisfied, we have immediately $\chsh>2$.

 If both conditions \eqref{eqn:nec_con1} and \eqref{eqn:nec_con2} are met, the maximum $\chsh$ achievable is given by
\begin{equation}
  \chsh_{\nu,\gamma} = 2\sqrt{\frac{[2c_1^2 - c_2(c_2+c_3)]^2}{4c_1^2-(c_2+c_3)^2} + c_2^2}.
\end{equation}

\subsection{Bell violation: imperfect case}

In the Bell test scenario we consider, we assume that the light field suffers from losses due to imperfect intensity transmission, $T_{\rm line}$. We also assume the atomic detection ($A_0$,$A_1$) and homodyne detection ($B_0$) have perfect detection efficiency, and that the photocounting ($B_1$) has efficiency $\eta$. Both the transmission and detection losses are modeled as beamsplitters with transmitivities $\sqrt{T_{\rm line}}$ and $\sqrt{\eta}$. In this case, the equations (\ref{eqn:c1_perf}-\ref{eqn:c3_perf}) become,
\begin{align}
  c_1 &= \half V\Tr\big((\kb{0}{\sqrt{T_{\rm line}}\alpha}+\kb{\sqrt{T_{\rm line}}\alpha}{0})B_0\big) \nn \\
  &= V \big( \frac{2}{\sqrt{\pi}} \int^b_{-b} dx \, e^{-x^2} \cos(\sqrt{2} x \sqrt{T_{\rm line}}|\alpha|) - e^{-\frac{T_{\rm line}|\alpha|^2}{2}}\big) \label{eqn:c1_imperf}
\end{align}
where $V = e^{-(1-T_{\rm line})\half{|\alpha|^2}}$, and
\begin{align}
  c_2 &= \Tr(\proj{0} B_1) = 1 \\
  c_3 &= \Tr(\proj{\sqrt{\eta T_{\rm line}} \alpha} B_1) = 2e^{-\eta T_{\rm line} \half{|\alpha|^2}} - 1. \label{eqn:c3_imperf}
\end{align}
The form of $c_1$ is derived using the convention that $\hat{x} = \frac{a+a\dagg}{\sqrt{2}}$. It should further be noted that we have used $\alpha$ completely imaginary in \eqnref{eqn:c1_imperf}. The intuition for this comes from the form of $B_0$, which is a measurement of the $X$ quadrature, and the form of $c_1$, which is the trace of $B_0$ and the ``off-diagonal'' terms $\kb{0}{\alpha} + \kb{\alpha}{0}$. Drawing the phase space distributions of $\ket{0}$ and $\ket{\alpha}$ would then show that using $i \alpha \in \reals$ must give the largest overlap of the projections of the distributions onto the $X$ quadrature. This also implies that the absolute phase of $\alpha$ is not important, but as long as the relative phases between the quadrature measurement and $\alpha$ is $\half \pi$, the $c_1$ term will be maximized.
  To satisfy conditions \eqref{eqn:nec_con1} and \eqref{eqn:nec_con2}, we need to maximize $c_1^2$. This can be done by looking for the maximum of \eqref{eqn:c1_imperf} over the binning parameter, $b$, which can be shown to occur when $b = \frac{\pi}{2\sqrt{2} |\alpha|}$. Further optimization over $\alpha$ is done numerically, and the results are shown in figure \ref{fig:ideal}. The maximum CHSH violation is 2.324, for $|\alpha| = 2.1$. For a line transmission of 52.2\%, and perfect detector efficiency, this scheme still achieves a violation of the CHSH inequality. We next show that it is always possible to achieve a violation for perfect line transmission, regardless of the detector efficiency.

\subsection{Bell violation with arbitrarily low photo-detection efficiency}
 \begin{figure}[ht]
	\centering
\includegraphics[width = \columnwidth]{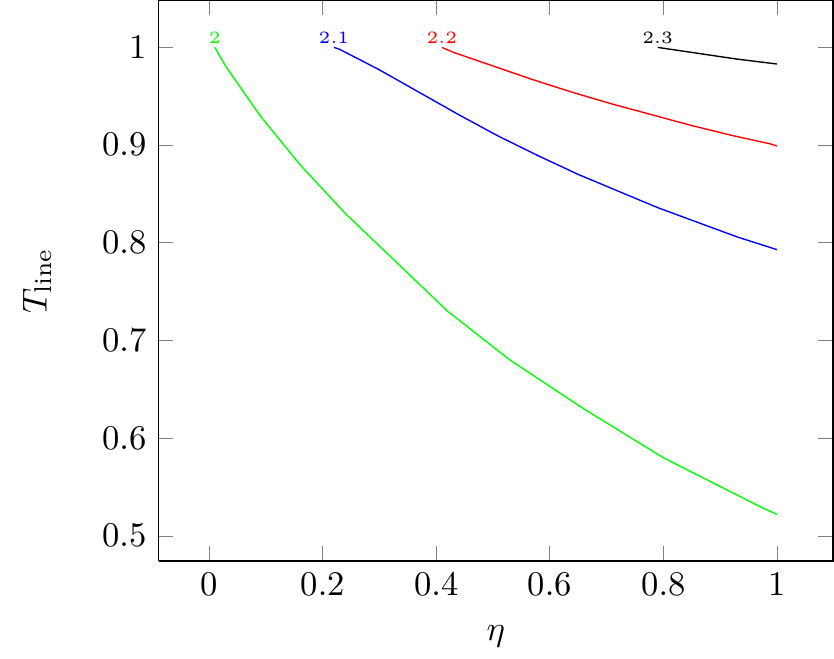}
\caption{(Color online) Contour plot of $\chsh$ vs the line transmission and the photocounting efficiency. The maximum CHSH violation is 2.324 for $|\alpha|=2.1$. Notice that these results are slightly better than Fig. 1 of Ref.~\cite{TeoAra13}.}\label{fig:ideal}
\end{figure}
\begin{thm} \label{thm:nice}
  If $T_{\rm line}=1$, we have
  \begin{equation}
    \chsh > 2,
  \end{equation}
  irrespective of the value of $\eta$.
\end{thm}
\begin{proof}
  We first note that since $c_2=1$, we need only satisfy condition \eqref{eqn:nec_con2} to show that $\chsh >2$. This can been seen since
  \begin{align}
    c_3 &\leq 1 = c_2, \\
    \implies c_3 (c_2+c_3) &\leq c_2 (c_2+c_3).
  \end{align}
  Also, for $T_{\rm line} = 1$, we have
  \begin{align}
    c_1 &= \big( \frac{2}{\sqrt{\pi}} \int^b_{-b} dx \, e^{-x^2} \cos(\sqrt{2} x |\alpha|) - e^{-\frac{|\alpha|^2}{2}}\big),\quad {\rm and} \label{eqn:proof_c1}\\
    c_3 &= 2e^{-\eta \half{|\alpha|^2}} - 1.
  \end{align}
  Then, for $|\alpha| \to \infty$, we have
  \begin{align}
    \frac{2}{\sqrt{\pi}}\int^b_{-b} dx \, e^{-x^2} \cos(\sqrt{2} x |\alpha|) &\approx \frac{2}{\sqrt{\pi}}\int^b_{-b} dx \,\cos(\sqrt{2} x |\alpha|), \\
    &= \sqrt{\frac{2}{\pi}} \frac{2}{|\alpha|},
  \end{align}
  where we have used $b = \frac{\pi}{2\sqrt{2} |\alpha|} \to 0$ for $|\alpha| \to \infty$, and the approximation that in the limit of small $b$, the term $e^{-x^2}$ can be approximated to unity. The condition \eqref{eqn:nec_con2} then becomes
\begin{equation}
  e^{-\eta \frac{|\alpha|^2}{4}} < \sqrt{\frac{2}{\pi}} \frac{2}{|\alpha|},
\end{equation}
where we have made the approximation that for $|\alpha|\gg 1$, $e^{-\half{|\alpha|^2}} \approx 0$ in \eqnref{eqn:proof_c1}. Since the left hand side of the inequality is exponentially decreasing in $|\alpha|$, and the right hand side is decreasing as $\inv {|\alpha|}$, for any $\eta$, $\exists$ some $|\alpha| \gg 1$ such that condition \eqref{eqn:nec_con2} holds. Thus, since $c_2=1$, for $T_{\rm line} = 1$ and \emph{any} $\eta$, there exists some $|\alpha|, \gamma$ and $\nu$, which gives $\chsh>2$.
\end{proof}

\subsection{Scenario with two homodyne measurements} \label{sec:doubleH}

Theorem \ref{thm:nice} shows that for any $\eta$, there exists a $\ket{\psi}$ which can violate the CHSH inequality. Moreover, the setup we have been considering is a highly asymmetric measurement setup, where the light field is subject to either photodetection or homodyne measurement. It is natural to ask if we can perform the Bell test in a symmetric manner, using two homodyne measurements on the photonic side, eliminating the need to consider inefficient photodetectors. \dani{In what follows we will show that a Bell violation of $\chsh\approx2.29$ can be reached in this scenario. Moreover, the minimum transmission needed to be able to violate CHSH is $T_{\rm line}\approx 0.68$}.

\subsubsection{Assuming perfect atomic measurements}
In the case of two homodyne measurements and perfect atomic measurements, the optimization is even simpler. We first redefine the measurement operators on the photonic side as
\begin{align}
  B_0 &= 2\int^b_{-b} dx\, \proj{x} - \one, \\
  B_1 &= 2\int^{\frac{|\alpha|}{\sqrt{2}}}_{-\infty} dp\,\proj{p} - \one, \label{eqn:B1_2H}
\end{align}
and consider only $\alpha$ imaginary as in \secref{sec:opt_methods}. This form of the $B_1$ measurement operator and binning is arbitrary, and we do not claim that it is optimal. The intuition for using this form of the $B_1$ measurement comes from the form of \eqnref{eqn:B1_n}. Notice that this operator discriminates the state $\ket{0}$, and the $\ket{\alpha}$ state. We thus choose \eqnref{eqn:B1_2H} to measure the $P$ quadrature to discriminate the states $\ket{0}$ and $\ket{\alpha}$, and choose a binning based on this intuition.

This choice of the $\hat{B}$ operators is also particularly convenient, since in this case, $c_2=-c_3$ (refer to \secref{sec:opt_methods}), and the function after optimization over $\gamma$ is
\begin{equation}
  \chsh_\gamma = 2 \sqrt{c_1^2 \sin^2(2\nu) + c_2^2},
\end{equation}
which is trivially maximized by $\sin^2(2\nu) = 1$ and $b = \frac{\pi}{2\sqrt{2} |\alpha|}$. Thus
\begin{equation}
  \chsh_{\nu,\gamma} = 2 \sqrt{c_1^2 + c_2^2}, \label{eqn:chsh_2H}
\end{equation}
where $c_1$ is given in \eqnref{eqn:c1_imperf}, and
\begin{equation}
  c_2 = \erf (\sqrt{\half {T_{\rm line}}}|\alpha|). \label{eqn:c2_2H}
\end{equation}
Optimizing \eqnref{eqn:chsh_2H} numerically over $|\alpha|$, we can plot $\chsh_{\nu,\gamma,|\alpha|}$ vs $T_{\rm line}$. \figref{fig:chsh_2h} summarizes the result.
It shows that the largest violation attainable is 2.29, and the minimum required transmission which still yields a violation to be $T_{\rm line}=0.678$. This measurement protocol can be directly compared to the one presented in Ref. \cite{sangouard11}, and one can see that there is both an improvement on the maximum violation and the minimum transmission (2.26 and ~78\%). One important remark is that in this measurement scenario, we no longer need to produce the state \eqref{eqn:target_state} exactly, and a state of the form
\begin{equation}
 \ket{\psi} = \cos \nu \ket{s,0+\beta} +\sin \nu \ket{g,\alpha + \beta}
\end{equation}
for any complex $\beta$ is also a possible candidate, since this corresponds to an appropriate shift in the binnings.
\begin{figure}[!ht]
	\centering
\includegraphics[width = \columnwidth]{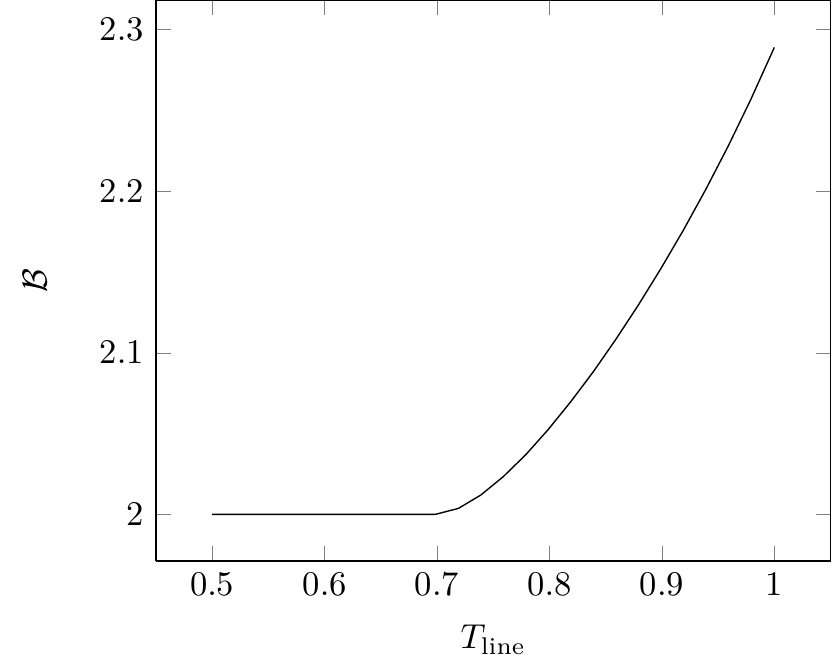}
\caption{$\chsh$ vs Transmission diagram for two homodyne measurements on the field. This figure shows that the maximum CHSH violation is 2.29, and the minimum transmission required for the violation of the CHSH inequality is $T_{\rm line}=0.678$. $\chsh$ is computed by numerically optimizing \eqnref{eqn:chsh_2H} over $|\alpha|$. }\label{fig:chsh_2h}
\end{figure}

\colin{We next investigate the behaviour of $|\alpha_{\rm opt}|$, the value of $|\alpha|$ which optimizes the CHSH violation as a function of the transmission $T_{\rm line}$. This is an important point, as the state production fidelity is usually inversely proportional to $|\alpha|$ (see for instance Ref.~\cite{TeoAra13}). As \figref{fig:2h_alpha_opt} shows, the typical range of the optimal $\alpha$ is $|\alpha_{\rm opt}|\in[2.2,3]$, which occurs for transmissions above 72\%.}
\begin{figure}[!ht]
	\centering
\includegraphics[width = \columnwidth]{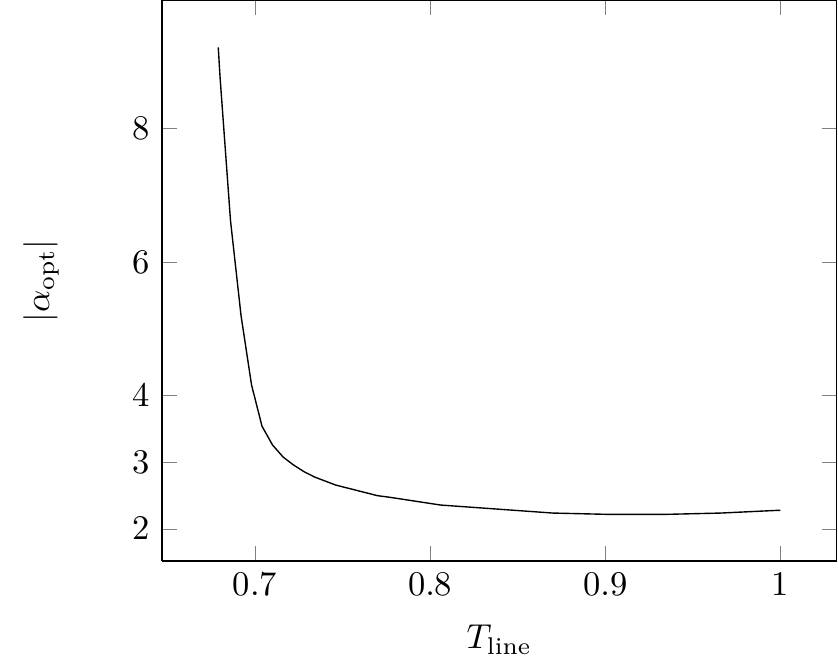}
\caption{$|\alpha_{\rm opt}|$ vs Transmission diagram. This diagram is the value of $|\alpha|$ which results from the numerical optimization of \eqnref{eqn:chsh_2H} over $|\alpha|$. The numerical optimization gives the maximum $|\alpha|$ as 9.22. However, the typical range of $|\alpha_{\rm opt}|$ is $\in [2.22,3]$ for $T_{\rm line}>0.72$   }\label{fig:2h_alpha_opt}
\end{figure}

The discussions above have assumed that atomic detection can be done with unit efficiency. However, the drawback of having high detection efficiencies in the atomic detection usually necessitates long detection times. Although schemes exist to implement fast and efficient atomic detection \cite{henkel2010highly}, such techniques might not always be available. In the next subsection, we treat the problem of inefficient atomic detection, and show that it quickly degrades the achievable CHSH violation.
\subsubsection{Inefficient atomic detection} \label{sec:ineffatoms}
To treat inefficient atomic detection, one assumes that with probability $\eta_{\rm a}$, everything proceeds as in the previous section, and with probability $1-\eta_{\rm a}$, the atomic measurement operators $A_0$ and $A_1$ have no effect, and can be modelled by the identity operation $\one$. Then, the CHSH quantity after optimization over the atomic measurement angle gives
\begin{equation}
  \chsh_\gamma = 2\Big( \eta_{\rm a} \sqrt{c_1^2 \sin^2(2\nu) + c_2^2} + (1-\eta_{\rm a}) c_2 \cos (2\nu) \Big),
\end{equation}
where $c_1$ and $c_2$ are once again given by Eqs. \eqref{eqn:c1_imperf} and \eqref{eqn:c2_2H}. The optimization over the angle $\nu$ can once again be done analytically, and the resulting expression numerically optimized over $\alpha$. \figref{fig:chsh_2h_aeta} summarizes the results of the final optimization.
\begin{figure}[ht]
	\centering
\includegraphics[width = \columnwidth]{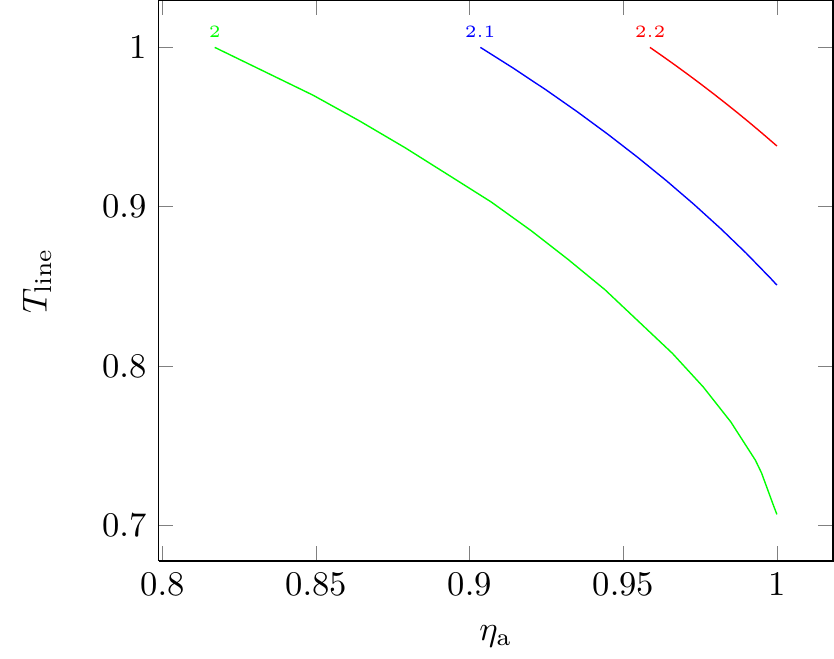}
\caption{(Color online) Contour plot of $\chsh$ vs the line transmission and the atomic detection efficiency for two homodyne measurements on the field. It shows that the minimum atomic detection efficiency is 0.817 for perfect line transmission.}\label{fig:chsh_2h_aeta}
\end{figure}
It shows that the atomic detection efficiency must at least be above $\eta_{\rm a} = 0.817$ to observe any CHSH violation, since it is the atomic detection efficiency for perfect line transmission. This result is not optimal, as a simple calculation would show that it does not reach the Eberhard bound, $\eta_{\rm a} T_{\rm line} = 2/3$ \cite{Ebe93}. In the next section, we show a possible scheme to circumvent inefficient atomic detection.

\subsection{Heralded preparation of entangled coherent states} \label{sec:atom_projection}

As a byproduct of our proposal, the setup proposed in Ref.~\cite{TeoAra13} to produce the state \eqref{eqn:target_state} can also be used to produce an entangled coherent states \cite{sanders_review}.
Consider the system prepared in the state \eqref{eqn:target_state}. If the atom is measured in the $\inv{\sqrt{2}}(\ket{s}\pm\ket{g})$ basis, and post-selected on the $\inv{\sqrt{2}}(\ket{s}+\ket{g})$ outcome, the resultant photonic state after a suitable displacement operation is,
\begin{equation}
  N(\alpha)(\cos\nu \ket{-\alpha} + \sin\nu \ket{\alpha}), \label{eqn:cat}
\end{equation}
where $N(\alpha)$ is a normalization factor, dependent on $\alpha$. This state is known as a coherent state superposition (or Schr\"{o}dinger's cat state) \cite{sanders_review}, and has been well studied in the literature, with much experimental progress in creating these states \cite{grangier_2009,sasaki_2008}. However, producing such states with values of $|{\alpha}|\geq 1.5$ has proven to be a big experimental hurdle. The state production protocol in Ref.~\cite{TeoAra13} thus provides an alternative route to achieving such states with a larger $|{\alpha}|$. One interesting thing about this state is that, by sending the state \eqref{eqn:cat} to a 50/50 beamsplitter, one can create entangled coherent states of the form
\begin{equation}
    N(\alpha)(\cos\nu \ket{-{\frac{\alpha}{\sqrt{2}}},-{\frac{\alpha}{\sqrt{2}}}} + \sin\nu \ket{{\frac{\alpha}{\sqrt{2}}},{\frac{\alpha}{\sqrt{2}}}}),
\end{equation}
which is reminiscent of the Bell states with polarization entanglement \cite{EnkHir01}. It is also worth noting, that such a route to producing a coherent state superposition with the help of atoms is not new, and a similar proposal has been previously studied in Ref.~\cite{duan_2005}.

\subsubsection{Splitting the cat}
Notice that the above operation of splitting the state \eqref{eqn:cat} using a single beamsplitter can, in principle, be repeated \emph{ad infinitum}, leading one to envision the case of $n$ beamsplitters as in \figref{fig:split_cats}.
\begin{figure}[ht!]
\centering
\includegraphics[width =  0.8\columnwidth]{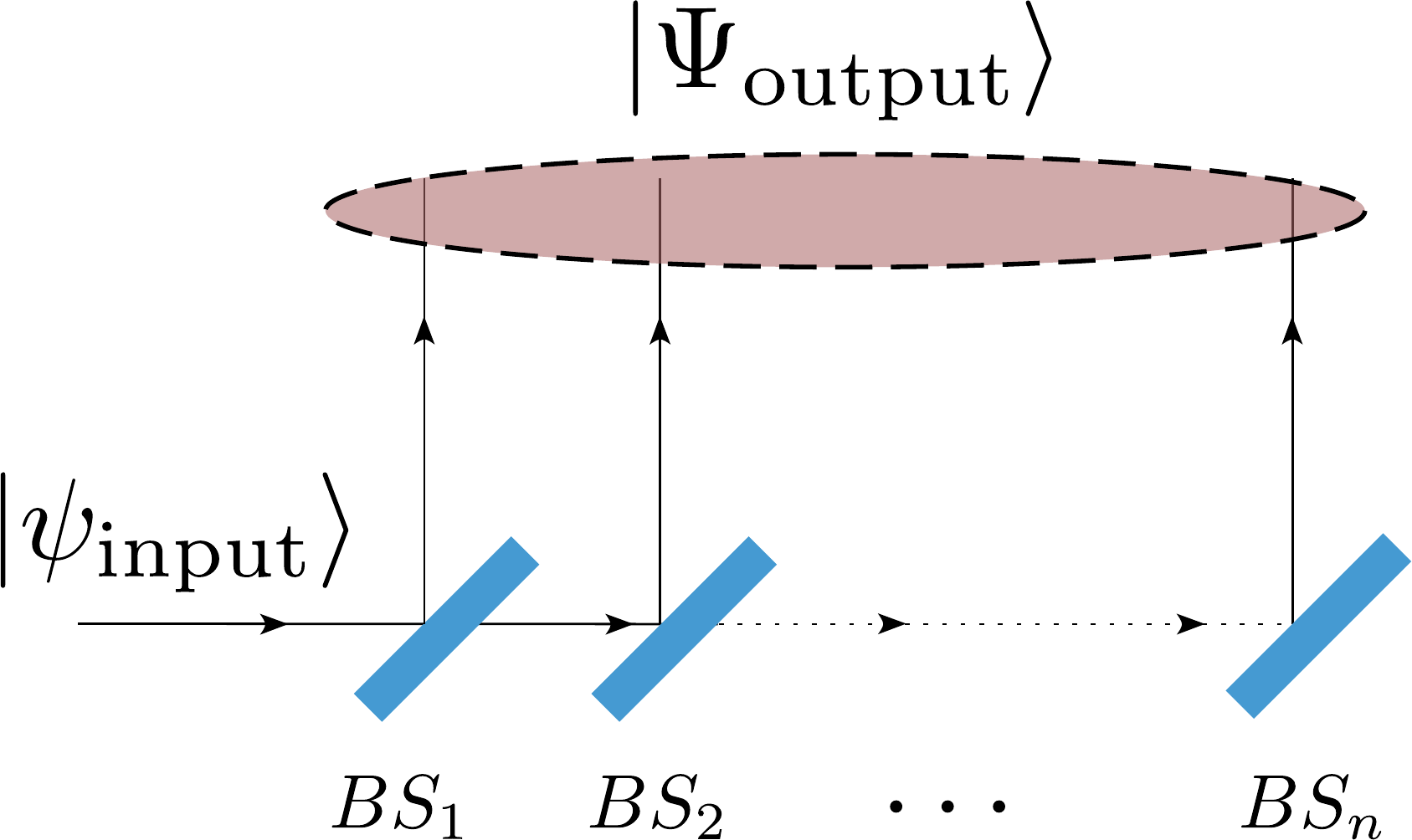}
\caption{(Color online) The above setup of $n$ beamsplitters creates a path-entangled, $n$-partite state $\ket{\Psi_{\rm output}}$. For input state $\ket{\psi_{\rm input}}$ given by \eqnref{eqn:cat}, the output state will be given by \eqnref{eqn:multi_cat}.} \label{fig:split_cats}
\end{figure}
Denoting the transmittivity and reflectivity of the $i^{th}$ beamsplitter as $t_i$ and $r_i$, with $t_i^2 + r_i^2 =1 $, we write the state $\ket{\Psi_{\rm output}}$ produced at the output given the input state \eqref{eqn:cat}, as
\begin{align}
  \ket{\Psi_{\rm output}} &= N(\alpha)(\cos\nu \ket{-f_1 \alpha,-f_{2} \alpha, \ldots\,-f_{N} \alpha } {} \nn \\
  &{} + \sin\nu \ket{f_1 \alpha,f_{2} \alpha, \ldots\,f_{N} \alpha }), \label{eqn:multi_cat}
\end{align}
where we have used the shorthand $f_k = r_k \Pi^{k-1}_{i=1} t_i$. Some simple algebra then shows that to create a state with equal amplitudes in each mode requires that the transmittivity of the $k^{th}$ and $k-1 ^{th}$ beamsplitters satisfy $t^2_k = 2- \inv{t^2_{k-1}}$, thus producing a state of the form
\begin{equation}
  \ket{\Psi} = N(\alpha) (\cos\nu \ket{-\ti \alpha,-\ti \alpha\ldots\,-\ti\alpha}+\sin\nu\ket{\ti \alpha,\ti \alpha\ldots\,\ti\alpha}).
\end{equation}

\section{Conclusion} \label{sec:conclude}


\colin{In this work, we have presented an improved analysis of the proposal to perform a loophole free Bell test in Ref.~\cite{TeoAra13}. We showed that the hybrid atom-photon entangled state and the set of measurements considered in Ref. \cite{TeoAra13} allows a violation of the CHSH inequality for \emph{vanishing} photocounting efficiency with perfect detection on all other measurements.

We also showed that we can perform the Bell test using only atomic and homodyne measurements, eliminating the need to consider inefficient photocounting, and showed CHSH violations down to a transmission of 67.8\%. Finally, we showed that the proposed system can be used as a state preparation device to produce a superposition of coherent states, useful for quantum information processing using continuous variables. }

\section*{Acknowledgements}
This work is funded by the National Research Foundation and the Ministry of Education of Singapore. We acknowledge discussions with Melvyn Ho, Mateus Ara\'ujo, Marco T\'ulio Quintino, Marcelo Terra Cunha and Marcelo Fran\c ca Santos.

\end{document}